\documentclass[]{interact}
\usepackage{diagbox}
\usepackage{caption}
\usepackage[compact]{titlesec}
\usepackage{graphicx}
\usepackage{mathtools}
\usepackage{upgreek}
\usepackage{comment}
\usepackage{algorithm,algpseudocode}
\usepackage{subfigure}
\usepackage{amsthm}
\usepackage{graphicx}
\usepackage{float}
\usepackage{amsfonts}
\usepackage{amsmath}
\usepackage{bm}
\usepackage{mathtools}
\usepackage{tabularx}
\usepackage{array}
\usepackage{breqn}
\usepackage{supertabular}
\usepackage{setspace}
\usepackage[labelsep=quad]{caption}
\newcolumntype{H}{>{\setbox0=\hbox\bgroup}c<{\egroup}@{}}
\usepackage{url}
\usepackage{ragged2e}  
\newcolumntype{Y}{>{\RaggedRight\arraybackslash}X} 
\usepackage[export]{adjustbox}
\usepackage{xcolor}
\algrenewcommand\algorithmicindent{0.5em}%
\newtheorem{theorem}{Theorem}
\newtheorem{lemma}{Lemma}
\newtheorem{corollary}{Corollary}[theorem]
\usepackage[inline]{trackchanges}
\addeditor{NG}

\begin{document}
\title{
Decentralized and Secure Generation Maintenance with Differential Privacy
}



\author{
\name{Paritosh Ramanan\textsuperscript{a,c}\thanks{CONTACT Paritosh Ramanan Email: paritoshpr@gatech.edu}, Murat Yildirim\textsuperscript{b}, Nagi Gebraeel\textsuperscript{c} and Edmond Chow\textsuperscript{a}}
\affil{\textsuperscript{a}School of Computational Science and Engineering, Georgia Institute of Technology; \textsuperscript{b}College of Engineering, Wayne State University; \textsuperscript{c}School of Industrial and Systems Engineering, Georgia Institute of Technology}
}

\maketitle

\begin{abstract}
Decentralized methods are gaining popularity for data-driven models in power systems as they offer significant computational scalability while guaranteeing full data ownership by utility stakeholders. However, decentralized methods still require sharing information about network flow estimates over public facing communication channels, which raises privacy concerns. In this paper we propose a differential privacy driven approach geared towards decentralized formulations of mixed integer operations and maintenance optimization problems that protects network flow estimates. We prove strong privacy guarantees by leveraging the linear relationship between the phase angles and the flow. To address the challenges associated with the mixed integer and dynamic nature of the problem, we introduce an exponential moving average based consensus mechanism to enhance convergence, coupled with a control chart based convergence criteria to improve stability. Our experimental results obtained on the IEEE 118 bus case demonstrate that our privacy preserving approach yields solution qualities on par with benchmark methods without differential privacy. To demonstrate the computational robustness of our method, we conduct experiments using a wide range of noise levels and operational scenarios. 
\end{abstract}
\section*{Nomenclature}
\textbf{Sets}:
\begin{center}
\begin{supertabular}{c lc}
$\mathcal{R}$ & The set of all regions\\[1mm]
$T$ & Operational planning horizon \\[1mm]
$M$ & Set of maintenance windows\\[1mm]
$\mathcal{N}_r,G_{r},\mathcal{U}_{r},\mathcal{V}_{r},\mathcal{I}_{r}$ & Neighboring regions, generators, boundary,  foreign and \\& internal buses of  region $r$\\[1mm]
$\mathcal{B}_r$ & $\mathcal{U}_r \cup \mathcal{V}_{r}$, Boundary, foreign buses of $r$\\[1mm]
$\mathcal{N}^b_r$ & Neighboring regions connected to bus $b \in \mathcal{U}_{r}$\\[1mm]
\end{supertabular}
\end{center}
\begin{center}
\begin{supertabular}{c lc}
$G^b_{r},\mathcal{U}^b_r,\mathcal{V}^b_{r},\mathcal{I}^b_r$ & Generators, boundary, foreign, internal  buses connected \\& to bus $b \in \mathcal{U}_{r} \cup \mathcal{I}_r$\\[1mm]
$G^d_{r}$ & Generators in region $r$ that require maintenance within \\& planning horizon\\[1mm]
$\mathcal{B}^b_r$ & $\mathcal{U}^b_r \cup \mathcal{V}^b_{r} \cup \mathcal{I}^b_r$, Neighboring buses of bus $b$\\
\end{supertabular}
\end{center}
\vspace{3mm}
\hspace{-3mm}\textbf{Decision Variables} (at $t\in T$ for $g \in G$):
\begin{center}
\begin{tabular}{c lc}
$y^{g}_t$ & The electricity dispatch variable\\[1mm]
$x^{g}_t \in \{0,1\}$ & The commitment variable \\[1mm]
$z^{g}_m \in \{0,1\}$ & The maintenance variable for $m \in M$\\[1mm]
$\theta^{b}_{t}$ & The phase angle at bus $b$\\[1mm]
$f^{uv}_{t}$ & Power flow from bus $u$ to $v$ such that $u \in \mathcal{U}_{r}$ and $v \in \mathcal{V}^u_{r}$\\[1mm]
$\lambda^{b}_t$ & The Lagrangian multiplier with respect to phase angles of \\& bus $b \in \mathcal{U}_{r} \bigcup \mathcal{V}_{r}$\\[1mm]
$\phi^{uv}_t$ & The Lagrangian multiplier with respect to flow from bus $u$ to $v$ \\& where $u \in \mathcal{U}_{r}$  and $v\in \mathcal{V}_{r}$ for region $r$\\[1mm] 
\end{tabular}
\end{center}
\textbf{Constants}:
\begin{center}
\begin{supertabular}{c lc}
$d^{g},c^{g}$ & The dispatch and commitment cost of $g$\\[1mm]
$P^{g}_{min}, P^{g}_{max}$ & Minimum and maximum capacity of $g$\\[1mm]
$\mu^g_U, \mu^g_D, R^g$ & Minimum up time, down time and ramp up, down constant for $g$\\[1mm]
$\delta^{b}_t$ & The demand at bus $b$ at $t \in T$\\[1mm]
$F^{uv}_{max}$ & Maximum capacity of line  connecting  buses $u$ and $v$ where $u \in \mathcal{U}_{r}$ \\& and $v \in \mathcal{V}^u_{r}$\\[1mm]
$\rho_{\theta},\rho_{f}$ & Penalty parameter for phase angles, flows \\[1mm]
$\Gamma(uv)$ & Phase angle conversion for line $uv$\\[1mm]
\end{supertabular}
\end{center}
\vspace{3mm}
\section{Introduction}
Planning problems are the cornerstone for efficient functioning of transmission systems in large scale power systems. Some examples of critical planning problems include economic dispatch \cite{chowdhury1990review}, optimal power flow\cite{capitanescu2011state}, unit commitment (UC) \cite{unit_commitment} and maintenance \cite{maintDereg,maintDereg2}. Optimal planning decisions are subject to operational and reliability constraints \cite{marwali1998integrated} which ultimately require solving large scale optimization problems. Recently, a growing body of literature has focused on the use of consensus driven, decentralized optimization strategies to address issues of computational scalability and data localization \cite{yang2013consensus,javad,ramanan2017asynchronous,asynch2019uc,xavier2020decomposable} pertaining to the power systems planning problem. Despite their success, decentralized methods require disclosure of network flow estimates to their peers in order to compute optimal decisions. Such disclosures typically take place over public facing communication channels such as the internet \cite{asynch2019uc} leading to privacy risks emanating from a malicious third party. To address this risk, in this paper, we develop a novel decentralized optimization framework that leverages differential privacy \cite{dwork2014algorithmic}, for protecting network flow estimates. 

Differential privacy is a widely used method to protect the privacy of datasets intended to be communicated through public domains \cite{dwork2014algorithmic,cortes2016differential}. Differential privacy driven approaches involve injecting a randomized noise in order to obfuscate the real underlying data record. The injected randomized noise can be designed so as to facilitate theoretical guarantees bounding the loss of privacy \cite{dwork2014algorithmic}. Differential privacy thereby ensures that the probability of extracting the real value from a noisy dataset by any external entity remains remarkably low. As a result, differential privacy forms an attractive option to preserve privacy of network flow values in decentralized planning problems. 

In order to demonstrate our framework, we develop a decentralized formulation of the generation maintenance problem \cite{wang2016stochastic} whose solution is critical to the scheduling of operations and maintenance over a designated planning window. Being a fundamental problem in power systems, generation maintenance is particularly susceptible to privacy and scalability issues. There are a number of unique aspects of the generation maintenance problem that makes it an interesting problem to study in our setting. 
First, it consists of binary decisions for generator maintenance across discrete time windows as well as hourly binary commitment decisions. Owing to binary decisions as well as a planning horizon of a week, the generation maintenance problem is large scale and mixed integer in nature. As widely-documented in decentralized optimization literature \cite{ramanan2017asynchronous}, mixed integer variables introduce significant challenges in model coordination. Second, we are focusing on sensor-driven generation maintenance which harnesses highly-sensitive asset-health data from generation assets. This information, if compromised, can lead to significant risks in asset safety and operational vulnerabilities. Third, the generation maintenance problem consists of multiple interdependent UC problems augmented with maintenance variables, making it a significantly more challenging problem than UC. It is evident, therefore, that the framework developed in this paper can be directly applied to the simpler, decentralized UC formulation as well. 

Our decentralized formulation decomposes the power network topology (i.e. spatial decomposition) into several regions which may represent various utility stakeholders or regional monitoring centers. To decompose the problem, we first relax the network flow constraints pertaining to transmission lines connecting two regions yielding regionally independent local subproblems. These constraints are dualized and incorporated into the objective function of these local subproblems to ensure coordination across subproblems. More specifically, the network flow estimates corresponding to the dualized constraints are balanced between neighboring regions through the iterative application of the Alternating Direction Method of Multipliers (ADMM) \cite{admm_boyd}. ADMM is  a key component of decentralized operational planning strategies \cite{javad,asynch2019uc}. 
In our framework, the ADMM methods will communicate a differentially private version of the phase angles across regions, from which the corresponding flow values will be estimated. 

Our strategy for differential privacy is based on the numerous benefits stemming from the relationship between the phase angles and flow. First, owing to their linear relation, a noise injection on the phase angles leads to a corresponding linear transformation being injected to the flow as well. Second, we note that in decentralized formulations of the planning problems, phase angles are primarily meant for computing the flow \cite{ramanan2017asynchronous,asynch2019uc}. Therefore, we can choose a noise to be injected on the phase angles such that its linear transformation leads to differential privacy guarantees on the corresponding flow values. 
If chosen carefully, the noise could also ensure privacy of flow values estimated from phase angle estimates emanating from different iterations and/or regions as well. Such a feature is significantly useful in the asymptotic sense, when the true phase angle estimates across multiple iterations and regions are very close to each other.

Further, we also note that improved convergence is all the more important in a differentially private setting for a dynamically evolving process (i.e. coordination mechanism causes the underlying phase angle and flow estimates to change through iterations). To discover this dynamic underlying convergence, we adopt an Exponentially Weighted Moving Average (EWMA) that processes the noisy phase angle estimates leading to faster convergence. In order to balance the trade-off between faster convergence and better solution quality, we employ the use of a regional control chart based on the Central Limit Theorem (CLT). Our control chart is applied on the consensus quantities estimated at every iteration on each region and is geared towards bringing an out of control process to in control. As a result, the control chart mechanism stabilizes the solution quality with respect to varying noise levels while retaining good convergence behavior.  

Our contributions in this paper can be summarized as follows:
\begin{itemize}
    \item We develop an ADMM based differentially private, decentralized planning framework for the generation maintenance problem. The mixed integer nature of the problem renders the ADMM application a significant challenge even without differential privacy. In our setting, this challenge is compounded by the use of differential privacy.
    \item We propose well-suited noise injection strategies that leverage the structure of the problem. Our approach injects an engineered noise at the level of the phase angles, that culminates in differential privacy of flow values between regions.
    \item We develop an EWMA-based mechanism to improve convergence at the presence of dynamically changing flow estimates. We evaluate the EWMA outputs within a CLT based control chart for stabilizing the solution quality. 
    \item We provide a High Performance Computing (HPC) driven implementation for simulating our framework under a diverse set of scenarios. 
\end{itemize}
Our experiments on the 8 and 12 region decompositions of the 118 bus case demonstrates that our approach is robust to a wide variety of noise scenarios and convergence limits. Extensive experiments demonstrate that the proposed approach provides stable solution quality that rivals its benchmark without any differential privacy.  

\section{Related Works}
In transmission system planning problems, the operational and reliability constraints rely on infrastructure data that is held locally by the various utility stakeholders \cite{ramanan2017asynchronous}. In order to solve planning problems, infrastructure data must be aggregated at a centralized location leading to privacy and cyber security risks \cite{asynch2019uc,javad}. In addition to revealing private and sensitive infrastructure data of the stakeholders, such a centralized computational model also leads to communication bottlenecks on the central location \cite{liu2018decentralized}.
In the context of mixed integer power system planning problems, decentralized unit commitment was first proposed in \cite{javad} as a means for obtaining optimal UC decisions for networks 
without central control. An asynchronous version of the decentralized UC framework was proposed in \cite{ramanan2017asynchronous} with the purpose of improving computational efficiency. An extended version of the asynchronous decentralized model was discussed in \cite{asynch2019uc}, which provided improved solution quality for a large scale problem setting. More recently, the work done in \cite{xavier2020decomposable} proposes a decentralized UC formulation using the Power Transfer Distribution Factor (PTDF) as a means to improve scalability. 

In this paper, we study generation maintenance problem that jointly identifies optimal maintenance and UC decisions. UC problems studied in \cite{javad,ramanan2017asynchronous,asynch2019uc,xavier2020decomposable} form a subproblem within our setting. 
In generation maintenance, there is rich literature in coordination mechanisms between generation companies and market operators in a deregulated market setting \cite{maintDereg,maintDereg2,ghazvini2012coordination}. Our focus is on integrated operations and maintenance problems that solve for optimal maintenance as well as operational schedules subject to network constraints \cite{marwali1998integrated,fu2007security,fu2009coordination}. Generation maintenance problems typically use periodic maintenance policies, which require fixed time-based requirements for generation assets based on manufacturer recommendations and field experience (i.e. yearly major overhaul requirements). In contrast, our approach uses sensor-data to conduct condition-based maintenance strategies as proposed in \cite{muratp1,muratp2}. We integrate asset failure risks obtained through sensor data, within a joint optimization of operations and maintenance decisions. We study the short term periodic maintenance problem setting as in \cite{wang2016stochastic}. Our approach could potentially be adapted to other maintenance problems as well as operational paradigms. In our setting, any compromise to information security can reveal network-wide vulnerabilities that can lead to cyber-physical attacks on power systems, and opportunities for market manipulation \cite{sun2018cyber,zhang2016inclusion}.

Most differential privacy approaches in power systems are geared towards public release of operational power flow (OPF) data for benchmarking purposes. They address issues such as: quantifying the dynamics between injected noise and topology  \cite{zhou2019differential}; injecting noise into the OPF constraint set while guaranteeing solution accuracy \cite{fioretto2018constrained}, perturbing transmission line parameters \cite{fioretto2019differential} and hiding sensitive load locations as well as values \cite{mak2019privacy,fioretto2019privacy}. Lastly, the authors in \cite{mak2019privacy2} adopt a distributed, differentially private, ADMM driven approach to solve the AC-OPF problem by perturbing the demand at each bus. 

In contrast to the above works, our framework is meant for utility stakeholders to schedule their local operations and maintenance subject to global consensus over network constraints. Incorporating differential privacy in decentralized formulations of mixed integer power system planning problems largely remains understated, and has not been studied in a generation maintenance setting. Due to unique challenges in generation maintenance (e.g. large scale and mixed integer nature of the problem, and dynamically changing phase angle values), existing differential privacy approaches do not scale to our problem setting, requiring us to develop novel approaches to address these challenges.

\section{Decentralized Short Term Maintenance}\label{sec:mpf}
Our differential privacy driven technique is motivated by the recent developments of decentralized computational methods in power systems. As a result, our privacy preserving problem formulation comprises three main components which are detailed in this section. First, we discuss a decentralized formulation that employs mixed integer optimization techniques to yield maintenance and operational decisions including hourly commitment schedules. Second, we discuss our novel differential privacy driven information exchange that is utilized for obtaining the ADMM balance of flow. Lastly, we present our privacy preserving optimization framework that incorporates EWMA as well as control charts for stable convergence and superior solution quality. 

\subsection{Decentralized Short Term Maintenance and Commitment}
We propose a decentralized formulation based on regional decomposition leading to the respective regional subproblems. From a practical standpoint, each region may denote a subsidiary of the utility company in a vertically integrated market or a utility company in a deregulated market. Therefore, every region is comprised of local generators and buses subject to its own operational constraints. 

We show the regional decomposition of a sample network with the help of Figure \ref{fig:sample_synchronous}.
\begin{figure}[!ht]
\centering
\includegraphics[width=0.4\textwidth,keepaspectratio]{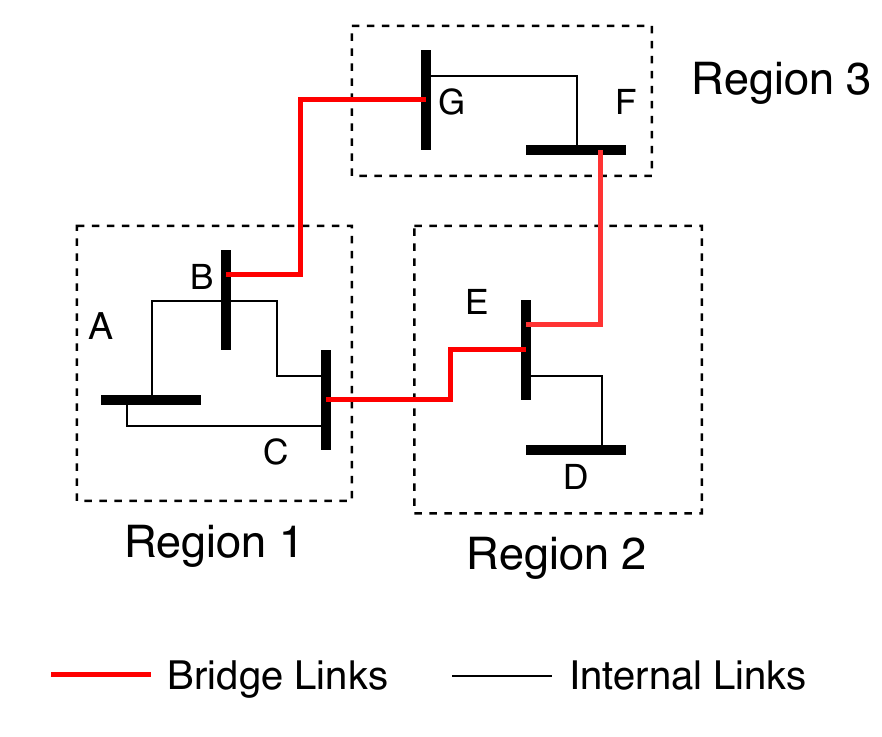}
\caption{Parition of Network topology into regions.}
\label{fig:sample_synchronous} 
\end{figure}
Our example network consists of 3 regions with boundary and foreign bus categorization for each region defined as follows:
\begin{itemize}
\item Region 1: $\mathcal{U}_{1} = \{B,C\}$, $\mathcal{V}_{1} = \{G,E\}$
\item Region 2: $\mathcal{U}_{2} = \{E\}$, $\mathcal{V}_{2} = \{F,C\}$
\item Region 3: $\mathcal{U}_{3} = \{G,F\}$, $\mathcal{V}_{3} = \{B,E\}$
\end{itemize}

The regional subproblem seeks to minimize the objective cost as represented by Problem \eqref{eq:OPT_OBJ} as follows. For simplicity we consider the vector form of the variables as necessary.

\begin{equation}\label{eq:OPT_OBJ}
\begin{aligned}
\underset{\bm{z},\bm{y},\bm{\lambda},\bm{\phi}}{\text{min}}\quad
& \mathcal{L}_r(\bm{\bar{\theta}}_k,\bm{\bar{F}}_k, \bm{\lambda}_k, \bm{\phi}_k) = \sum\limits_{ t \in T}\sum\limits_{g \in G_r}D_gy^g_t + C_gx^g_t  \\
&+ \sum\limits_{ m \in M}\sum\limits_{g \in G^d_r} K_g^m z_g^m\\
&+ \sum\limits_{t \in T}\sum\limits_{b \in \mathcal{B}_{r}}	\big[\lambda^{b}_t		|\theta^b_t - \bar{\theta}^b_t|	+	\frac{\rho_{\theta}}{2} (\theta^b_t - \bar{\theta}^b_t)^2 \big] \\
&+ \sum\limits_{t \in T}\sum\limits_{u \in \mathcal{U}_{r}}\sum\limits_{v \in \mathcal{V}^u_{r}}	\big[\phi^{uv}_t|f^{uv}_t -\bar{f}^{uv}_t|+ \frac{\rho_f}{2} (f^{uv}_t - \bar{f}^{uv}_t)^2 \big] 
\end{aligned}
\end{equation}
The regional objective function represented by Problem \eqref{eq:OPT_OBJ} consists of a dispatch cost component (the term with $D_g y^g_t$) a commitment cost component (the term with $C_g x^g_t$) and a dynamic maintenance cost component (the term with $ K_g^m z_g^m$) that is stored locally. The dynamic maintenance cost $K_g^m$ is dynamically evaluated based on sensor-driven predictions on the remaining life of generator $g$. For more information on this cost factor, we refer the reader to \cite{muratp1,muratp2}. In addition, the objective function also includes ADMM penalty terms imposed to balance flow estimates among neighboring regions. Flow estimates are iteratively balanced across transmission lines between neighboring regions through an iterative process. After every local solve of Problem \eqref{eq:OPT_OBJ}, the fresh estimates of phase angles are shared with neighbors in order to balance flows. Based on estimates received from neighbors, a consensus quantity can be estimated for flow as well as phase angles denoted by $\bar{f}, \bar{\theta}$ respectively.
 
Commitment, production and maintenance decisions are computed based on locally held constraints. Each maintenance window is comprised of several operational decision points such that,
\[T =\bigcup\limits_{m=1}^{|M|}T_m,\text{ where, }T_m = \Bigg\{t | t\in \Bigg[\frac{m|T|}{|M|}\ldots\frac{(m+1)|T|}{|M|}\Bigg]\Bigg\} \]
The regional subproblem for the joint operations and maintenance optimization is subject to a number of local constraints for $t\in T_m$, $m\in M$ represented by the set $Q^r$:
\begin{subequations}\label{eq:OPT}
\allowdisplaybreaks[1]
\begin{align}
&x^g_t \leq 1- z^g_m \quad \forall g \in G_r \label{eq:op_mt_cpl}\\[1mm]
&P^g_{min}x^{g}_t \leq y^g_t \leq P^g_{max}x^{g}_t,
\quad \forall g \in G_r & \label{eq:sq1}\\[1mm]
&-\pi^g_{Dt} \leq x^g_t -x^g_{t-1} \leq \pi^g_{Ut}, 
\ \forall g \in G_r & \label{eq:sq2}\\[1mm]
&-R^g \leq y^g_t - y^g_{t-1} \leq R^g,  
\ \forall g \in G_r & \label{eq:sq3}\\[1mm]
\begin{split} & \ \Gamma(uv)(\theta^u_t - \theta^v_t) = f^{uv}_t, \ \forall u \in \mathcal{U}_{r}, \forall v \in \mathcal{V}^u_{r}\\
\end{split} & \label{eq:sq4}\\[1mm]
\begin{split} & -F^{uv}_{max}\leq \Gamma^{uv}(\theta^u_t - \theta^v_t) \leq F^{uv}_{max}, 
\ \forall u \in \mathcal{U}_{r} \cup \mathcal{I}_r,\forall v \in \mathcal{B}^u_{r}\end{split} & \label{eq:sq5}\\[1mm]
\begin{split} & \sum\limits_{\forall g \in G^u_{r}}y^g_t - \delta^{u}_t +\psi_t^u = \sum\limits_{\forall v \in \mathcal{B}^u_{r}}[\Gamma^{uv}(\theta^{u}_t - \theta^{v}_t)],\forall u \in \mathcal{U}_{r} \cup \mathcal{I}_r\end{split}\label{eq:sq7}\\[1mm]
\begin{split} & \sum\limits_{\forall i\in U_t} \pi^g_{Ui} \leq x^g_t \leq 1-\sum\limits_{\forall i \in D_t} \pi^g_{Di}, \ \forall g \in G_r, U_t=[t-\mu^g_U+1,t], D_t = [t-\mu^g_D+1,t] \end{split} \label{eq:sq8}
\end{align}
\end{subequations}
Constraint \eqref{eq:op_mt_cpl} ensures that generators placed under maintenance do not have any production, where $M(t)$ represents the maintenance window corresponding to the operational decision point $t$. Constraint \eqref{eq:sq1} enforces limits on each generator's maximum and minimum production levels. Constraints \eqref{eq:sq2} and \eqref{eq:sq8} enforce minimum up and down-time for each generator. Constraint \eqref{eq:sq3} enforces generator ramping limits. Equation \eqref{eq:sq4} enforces the linear relationship between flows and their respective phase angles. Constraint \eqref{eq:sq5} limits the transmission line capacity. Equation \eqref{eq:sq7} balances demand at each bus with corresponding generation and network flow. Equations \eqref{eq:sq4}-\eqref{eq:sq7} enforce network flow constraints globally.

In addition to these constraints, we also enforce that every degraded generator $g \in G^d_r$ in region $r$ is maintained within the planning horizon:
\begin{equation}\label{eq:maint_m}
\sum_{m \in \mathcal{M}} z^g_m = 1\quad \forall g \in G_r^d
\end{equation}

\section{Differential Privacy For Decentralized Planning}
Decentralized optimization frameworks for power system planning problems rely on ADMM based, iterative, flow and phase angle balancing between neighbor regions \cite{ramanan2017asynchronous,asynch2019uc,xavier2020decomposable}. In order to converge to the global optimum, regions share phase angle estimates with neighbors which can in turn also be used to estimate flow using Equation \eqref{eq:sq4} \cite{ramanan2017asynchronous,asynch2019uc}. Based on the phase angle and flow estimates received, regions compute consensus quantities $\bar{\theta},\bar{f}$ as shown in Problem \eqref{eq:OPT_OBJ}. These consensus quantities are critical for updating the Lagrangian duals $\lambda,\phi$ and therefore strongly influence global convergence to the optimal solution. In this paper, we focus on computing globally optimal, short term maintenance and operational decisions for every region in a decentralized fashion while employing differential privacy for protecting the flow. 

\subsection{Differential Privacy Primer}
Before delving into a detailed discussion of our framework, we briefly review key theorems pertaining to differential privacy that are vital to our framework.
\begin{theorem}\label{thm1}
\textbf{Definition} \cite{dwork2014algorithmic} : A randomized mechanism $\mathcal{M}$ with domain $\mathcal{R}$ and sensitivity $\omega >0$ is said to preserve $\epsilon$-differential privacy for some $\epsilon\geq 0$ 
$\forall x,x' \in \mathcal{R}$ with $||x-x'||_1 \leq \omega$, if the following relation holds:
\[Pr(\mathcal{M}(x) \in \mathcal{R})\leq e^{\epsilon}Pr(\mathcal{M}(x') \in \mathcal{R})\]
\end{theorem}
In other words, Theorem \ref{thm1} ensures that the probability of computing the exact distance of the true value $x$ from its perturbation $\mathcal{M}$ is low. The parameter $\epsilon$ represents a privacy budget, with a smaller value favoring a higher degree of privacy. The sensitivity parameter $\omega$ ensures obfuscation of values close to each other while maintaining the relative difference of values far apart. 

\begin{theorem}\label{thm2}
\textbf{Post Processing Immunity} \cite{dwork2014algorithmic} : Given a mechanism $\mathcal{M}$ that preserves $\epsilon$-differential privacy, then for any function $g$, the functional composition $g \circ \mathcal{M}$ also preserves $\epsilon$-differential privacy.
\end{theorem}
An important result of differential privacy pertains to post processing immunity encapsulated in Theorem \ref{thm2}. The post processing immunity implied by Theorem \ref{thm2} means that once differential privacy has been applied on any element of the domain $\mathcal{R}$, no further privacy can be lost with application of any arbitrary function by a third party.

\begin{theorem}\label{thm3}
\textbf{Adaptive Composition} \cite{cortes2016differential} : Given mechanisms $\mathcal{M}_1,\mathcal{M}_2$, which ensure $\epsilon_1,\epsilon_2$ differential privacy respectively, a mechanism $\mathcal{M}(x) = \mathcal{M}_2(x,\mathcal{M}_1(x))$ also preserves differential privacy.
\end{theorem}
Theorem \ref{thm3} provides a critical result that is especially useful for developing a decentralized, iterative optimization framework, wherein existing results depend on consensus estimates obtained from the previous iteration.

Lastly, in Theorem \ref{thm4}, we describe the Laplacian mechanism, which is one of the most commonly used techniques for ensuring differential privacy.
\begin{theorem}\label{thm4}
\textbf{Laplacian mechanism} \cite{dwork2014algorithmic} : Given a function $g$ with domain $\mathcal{R}$ and $\omega>0$, a mechanism $\mathcal{M}(x) = g(x)+w$, where $x\in \mathcal{R}$ and $w\sim Lap(0,\frac{\omega}{\epsilon})$ is $\epsilon$-differentially private.
\end{theorem}
Theorem \ref{thm4} guarantees that the mechanism $\mathcal{M}$ based on the Laplacian distribution with zero mean and standard deviation $\frac{\omega}{\epsilon}$ preserves $\epsilon$-differential privacy. 
\subsection{Differential Privacy for Flow}
In order to protect regional flow values with differential privacy, we exploit the linear relationship between phase angles and the corresponding flow variables. As a result, we propose Theorem \ref{thm5} which relies on an important property of Laplace distributions stated in Lemma \ref{lem1} \cite{kotz2012laplace}.
\begin{lemma}\label{lem1}  
Given two random variables $X,Y \sim Exp(\omega)$, the random variable $X-Y \sim Lap (0,\omega)$.
\end{lemma} 
\begin{theorem}\label{thm5}
Given a transmission line across $b_1b_2$ between regions $r_1,r_2$, such that $b_1\in\mathcal{U}_{r_1}, b_2 \in \mathcal{V}^{b_1}_{r_1}, b_2\in\mathcal{U}_{r_2}, b_1 \in \mathcal{V}^{b_2}_{r_2}$, a mechanism $\mathcal{T}$ given by 
\[\mathcal{T}(\theta^{b_i}_t) = \theta^{b_i}_t + \alpha^{b_i}_t, \textit{ where, } \alpha^{b_i}_t \sim Exp\Big(\frac{\omega}{|\Gamma(b_1b_2)|\epsilon}\Big) \textit{ , } i\in{1,2}\]
preserves $\epsilon$-differential privacy of flow $f^{b_1b_2}_t$ defined as $\mathcal{M}'(\theta_t^{b_1},\theta_t^{b_2}) = \Gamma(b_1b_2) \Big(\mathcal{T}(\theta^{b_1}_t) - \mathcal{T}(\theta^{b_2}_t)\Big)$.
\end{theorem}
\begin{proof}
Consider a Laplacian distribution $Lap(0,\tilde{\omega}/\epsilon)$, where $\tilde{\omega} = \omega(\Gamma(b_1b_2))^{-1}$. Substituting $\mathcal{T}(\theta^{u}_t)$ in Equation \eqref{eq:sq4}, we have 
\begin{align}\label{eq:flowdp}
\mathcal{M}'(\theta_t^{b_1},\theta_t^{b_2}) &= \Gamma(b_1b_2)\Big(\mathcal{T}(\theta^{b_1}_t) - \mathcal{T}(\theta^{b_2}_t)\Big) \\  &=  \Gamma(b_1b_2)(\theta^{b_1}_t - \theta^{b_2}_t) + \psi^{b_1,b_2}_t
\end{align}
where \eqref{eq:flowdp}, $\psi^{b_1,b_2}_t = \Gamma(b_1b_2)(\alpha^{b_1}_t - \alpha^{b_2}_t)$. Based on Lemma \ref{lem1}, we can claim that the random variable given by $\frac{\psi^{b_1,b_2}_t}{\Gamma(b_1b_2)} = (\alpha^{b_1}_t - \alpha^{b_2}_t)$ follows $Lap(0,\tilde{\omega}/\epsilon)$. Owing to the Laplace distribution being symmetric, we can state that the random variable $\psi^{b_1,b_2}_t \sim Lap(0,\omega/\epsilon)$. 
Since $\Gamma(b_1b_2)(\theta^{b_1}_t - \theta^{b_2}_t) = f^{b_1,b_2}_t$ denotes the real flow and $\psi^{b_1,b_2}_t \sim Lap(0,\omega/\epsilon)$ is Laplace distributed, Theorem \ref{thm4} concludes that mechanism $\mathcal{M}'$ preserves the $\epsilon$-differential privacy of flow.
\end{proof}

An important consequence of Theorem \ref{thm5} is that the phase angle values with an exponential perturbation directly lead to imposing the $\epsilon$-differential privacy for the flow.  Therefore, our iterative scheme is based on sharing the phase angles and by extension the differentially private flow values as well. 
\begin{corollary}\label{corl1}
Given $\theta^{b_1,k,r_1}_t, \alpha^{b_1,k,r_1}_t$ for bus $b_1$ at iteration $k$ and region $r_1$, and $\theta^{b_2,j,r_2}_t, \alpha^{b_2,j,r_2}_t$ for bus $b_2$ at iteration $j$ and region $r_2$, respectively; the mechanism $\mathcal{M}'(\theta^{b_1,k,r_1}_t,\theta^{b_2,j,r_2}_t)$ still preserves the $\epsilon$-differential privacy of the corresponding flow.

\end{corollary}
\begin{proof}
We note that 
\begin{equation}
    \mathcal{T}(\theta^{b_1,k,r_1}_t) - \mathcal{T}(\theta^{b_2,j,r_2}_{t}) = (\theta^{b_1,k,r_1}_t - \theta^{b_2,j,r_2}_{t}) + (\alpha^{b_1,k,r_1}_t - \alpha^{b_2,j,r_2}_{t})
\end{equation}
Following a similar reasoning as presented in Theorem \ref{thm5}, Lemma \ref{lem1} indicates that the noise given by $\alpha^{b_1,k,r_1}_t - \alpha^{b_2,j,r_2}_{t}$ follows a Laplacian distribution ensuring $\epsilon$-differential privacy. From Theorem \ref{thm2}, it also follows that $\Gamma(b_1 b_2)(\mathcal{T}(\theta^{b_1,k,r_1}_t) - \mathcal{T}(\theta^{b_2,j,r_2}_{t}))$ also remains differentially private.
\end{proof}
Corollary \ref{corl1} establishes the fact that phase angle values derived from a combination of historically observed values cannot be used to infer actual flow values. Thus, Corollary \ref{corl1} is especially useful in an asymptotic sense when the real phase angle estimates across multiple iterations and regions might be close. Finally, we note that the indices $1$ and $2$ were used for ease of exposition, and the theorem applies for flow estimates of any transmission line. 

\section{Algorithm Design for Decentralized Differential Privacy }
We employ Theorem \ref{thm5} in order to construct our decentralized algorithm with differential privacy. Our algorithm consists of two key components pertaining to improved convergence and added stability. For improving convergence we use an EWMA based consensus technique to balance phase angles and flow. On the other hand, for added stability, we propose the use of CLT control charts. 
\subsection{EWMA based Consensus Mechanism for ADMM}
At iteration $k$, each region shares the noisy phase angles $\hat{\theta}^{b,k}_t = \theta^{b,k}_t + 2\alpha^{b,k}_t$, $\forall b \in \mathcal{U}_{r}\bigcup \mathcal{V}_r$ with its neighbors. The noisy phase angle estimates are utilized to compute the noisy flow estimates $\hat{f}^{uv,k}_t$, $\forall u \in \mathcal{U}_{r}, \forall v \in \mathcal{V}^u_{r}$ based on Theorem \ref{thm5}. 

Based on the received values from neighbor $r' \in \mathcal{N}_r$, we estimate the two consensus terms for transmission line $uv$ separately leading to the \emph{intermediate flow},\emph{intermediate phase angle} denoted by $\bar{f}^{uv,k}_t$, $\bar{\theta}^{b,k}_t$ respectively. However, for added stability we apply an Exponentially Weighted Moving Average (EWMA) to the received values from the neighbor. EWMA leads us to Equations \eqref{eq:tintm1} to \eqref{eq:fintm3} with the mixing factor denoted by $\eta$.
Specifically, for phase angles of bus $b \in \{u,v\}$ corresponding to transmission line $uv$ we have,
\begin{gather}
\bar{\theta}^{b,k}_t = \frac{\theta^{b,k}_t + \tilde{\theta}^{b,k,r'}_t}{2} \label{eq:tintm1}\\
\tilde{\theta}^{b,k,r'}_t = \eta(\hat{\theta}^{b,k,r'}_t) + (1-\eta)(\tilde{\theta}^{b,k-1,r'}_t) \label{eq:tintm2}\\
\tilde{\theta}^{b,0,r'}_t = \hat{\theta}^{b,0,r'}_t \label{eq:tint3}
\end{gather}
Similarly for flow on transmission line $uv,\forall u \in \mathcal{U}_{r}, \forall v \in \mathcal{V}^u_{r}$ we have,
\begin{gather}
\bar{f}^{uv,k}_t = \frac{f^{uv,k}_t + \tilde{f}^{uv,k,r'}_t}{2} \label{eq:fintm1} \\
\tilde{f}^{uv,k,r'}_t = \eta(\hat{f}^{uv,k,r'}_t) + (1-\eta)(\tilde{f}^{uv,k-1,r'}_t) \label{eq:fintm2}\\
\tilde{f}^{uv,0,r'}_t = \hat{f}^{uv,0,r'}_t \label{eq:fintm3}
\end{gather}
We update the Lagrangian multipliers as follows
\begin{equation}\label{eq:lupdt}
\lambda^{b,k}_t = \lambda^{b,k-1}_t+ \rho_{\theta}(\theta^{b,k}_t - \bar{\theta}^{b,k}_t) \quad  \forall b \in \mathcal{B}_r, \forall t \in T
\end{equation}
\begin{equation}\label{eq:fupdt}
\phi^{uv,k}_t = \phi^{uv,k}_t + \rho_f(f^{uv,k}_t - \bar{f}^{uv,k}_t) \quad \forall u \in \mathcal{U}_{r}, \forall v \in \mathcal{V}_{r}, \forall t \in T
\end{equation}
The optimization model given by (\ref{eq:OPT}) describes a Mixed-Integer Quadratic Problem (MIQP) which solves for the maintenance and operations in a decentralized manner. Since we have the presence of binary variables in $z$, our problem is \emph{non-convex}. As a result, it becomes much harder than traditional convex schemes to achieve convergence in a decentralized manner. Recent works have demonstrated the successful application of ADMM for solving decentralized non-convex problems. The maintenance cost is fed as input to the data and is derived from the work done in \cite{muratp2}. 

The Lagrangian terms in the model serve as \emph{penalties} for deviating from a position of balance. Convergence occurs when these terms become small enough such that the optimization problem given by (\ref{eq:OPT}) becomes mathematically equivalent to that of a \emph{centralized} problem as described in \cite{muratp2}

\subsection{CLT Control Chart based Convergence Criteria}
In order to bolster the robustness of overall solution quality against the noise, we employ a CLT based control chart as our convergence criteria. Our use of the CLT driven control chart is driven by the symmetric nature of the Laplacian distribution. Consider $\Uptheta$ defined as,
\begin{gather}
\Uptheta = \sum\limits_{k=1}^{S_w}(\theta^{u,k,r}_t + \alpha^{u,k,r}_t - \theta^{u,k,r'}_t - \alpha^{u,k,r'}_t)
\end{gather}
where $S_w$ are the number of iterations which form one point on the CLT.
Asymptotically, we expect the real phase angles from different regions obtained at iteration $k$ to match. Applying CLT on their corresponding noise leads us to require that $\Uptheta \sim N(0,\frac{2\tilde{\omega}^2}{S_w})$. Under CLT, the phase angle control process triggers an alarm whenever $\Uptheta> \sqrt{\frac{2\tilde{\omega}^2}{S_w}}$ or $\Uptheta< -\sqrt{\frac{2\tilde{\omega}^2}{S_w}}$. A single value of $\Uptheta$ forms one point on the control chart for the respective bus. We assume that our stopping criteria triggers local convergence only when $S_p$ points on the control chart lead to a single alarm. This implies that local convergence can only be certified in multiples of $S_w\cdot S_p$ iterations.
\subsection{Decentralized Maintenance and Operations Algorithm}\label{sec:malg}
Our decentralized and differentially private algorithm relies on two key components pertaining to the local optimization solves on every region and the peer-to-peer communication scheme. For notational simplicity, we denote $\Delta = \{\bm{\bar{\theta}}_k,\bm{\bar{F}}_k, \bm{\lambda}_k, \bm{\phi}_k\}$. 
\subsubsection{Local Optimizer}
The subroutine \texttt{OptSolve} represents the local optimization that occurs at every region. Specifically, our regional solver consumes the constraint set represented by $Q$ in addition to the objective function \eqref{eq:OPT_OBJ} derived from latest consensus variables as well as their corresponding multipliers. It returns the latest estimates pertaining to the commitment, production, maintenance decisions as well as unperturbed phase angles and flow estimates.
\begin{algorithmic}\label{alg:lmtopt}
\Function{OptSolve}{$\mathcal{L}_r(\Delta),\bm{Q}$}
    \State $\{\bm{x},\bm{y},\bm{z},\bm{\theta},\bm{f}\}\leftarrow{\text{min }} \mathcal{L}_r(\Delta)$ subject to $Q$
    \State \Return $\{\bm{x},\bm{y},\bm{z},\bm{\theta},\bm{f}\}$
\EndFunction
\end{algorithmic}
\subsubsection{Differential Privacy driven Communication scheme}
The subroutine \texttt{DPCommunicate} represents the peer-to-peer to communication scheme that adopts the phase angle and flow based differential privacy scheme given in Theorem \ref{thm5}. Exponentially perturbed phase angle estimates are shared with neighbors and consensus quantities are computed at every iteration. Local convergence occurs when the phase angle residual values are below the primal and dual tolerance ($\beta_p,\beta_d$) respectively and the number of alarms $\kappa$ in the region is less than $|\mathcal{U}_r \cup \mathcal{V}_r|$.
\begin{algorithmic}
\Function{DPCommunicate}{$k,\Delta^{k-1},\bm{\theta}^k,\bm{f}^k$}
\State send $\mathcal{M}(\bm{\theta}^{b,k})$ to all regions $r',$ $\forall r' \in \mathcal{N}_r, b\in \mathcal{B}_{r'}$
        \State receive $\bm{\tilde{\theta}}^{b,k}$ from all regions $r',$ $\forall r' \in \mathcal{N}_r, b\in \mathcal{B}_{r'}$
        \State compute $\bm{\bar{\theta}}_k,\bm{\bar{f}}_k, \bm{\lambda}_k, \bm{\phi}_k$ based on Equations \eqref{eq:tintm1}-\eqref{eq:fupdt}
        \If {$||\bm{\theta}^{k}-\bm{\bar{\theta}}^k||<\beta$ $\&$ $||\bm{\bar{\theta}}^{k}-\bm{\bar{\theta}}^{k-1}||<\beta$ $\&$ $\kappa<|\mathcal{U}_r\cup \mathcal{V}_r|$}
        	\State set local convergence to true
        	\State if local convergence is true $\forall r \in \mathcal{R}$ then $\Upomega\leftarrow 1$
	    \EndIf
	    \State \Return $\{\Delta^{k},\Upomega\}$
\EndFunction
\end{algorithmic} 
\subsubsection{Regional Solver}
The subroutine \texttt{DecentDPOpt} represents the regional solver which iteratively invokes the local optimizer followed by a round of differentially private message exchange with the neighbors.
\begin{algorithmic}\label{alg:ddpopt}
\Function{DecentDPOpt}{$\Delta,\bm{Q}^r$}
\State $k\leftarrow 0$, $\Delta_0\leftarrow \Delta$,
\State set global convergence value $\Upomega\leftarrow0$
\While {$\Upomega\neq0$}
\State $k\leftarrow k+1$
\State $\{\bm{x}^{k},\bm{y}^{k},\bm{z}^{k},\bm{\theta}^k,\bm{f}^k\}$$\leftarrow$$\Call{OptSolve}{\mathcal{L}_r(\Delta^{k-1}),\bm{Q}^r}$
\State $\{\Delta^{k},\Upomega\}\leftarrow\Call{Communicate}{k,\Delta^{k-1},\bm{\theta}^k,\bm{f}^k}$
\EndWhile
\State \Return $\{\bm{x}^k,\bm{y}^{k},\bm{z}^{k},\Delta^k\}$
\EndFunction
\end{algorithmic}
\begin{algorithm}
\caption{Short term Decentralized Maintenance and Operations Algorithm}\label{alg:syncd}
\begin{algorithmic}
\State $\{\bm{x},\bm{y},\bm{z},\Delta\}_{R}$$\leftarrow$$ \Call{DecentDPOPT}{\Delta_0,\bm{Q}^r_{relax}}$
\State$\{\bm{x},\bm{y},\bm{z},\Delta\}$$\leftarrow$$ \Call{DecentSGOpt}{\Delta_{R},\bm{Q}^r_{bin}}$
\end{algorithmic}
\end{algorithm}
We use the regional solver to construct our decentralized algorithm represented in Algorithm \ref{alg:syncd} which comprises of two phases. In the first phase, we obtain convergence using a relaxation of the binary commitment and maintenance variables represented by the constraint set $Q^r_{relax}$. We utilize the phase and flow balance attained from the convergence of the relaxation to jump start the next phase involving binary commitment and flow variables represented by $Q^r_{bin}$. Such a two phased method has been demonstrated as one of the feasible methods for decentralized, mixed integer convergence in \cite{ramanan2017asynchronous,asynch2019uc}.
\begin{figure*}[!ht]
\includegraphics[trim=140 5 60 35,clip,width=\textwidth,keepaspectratio]{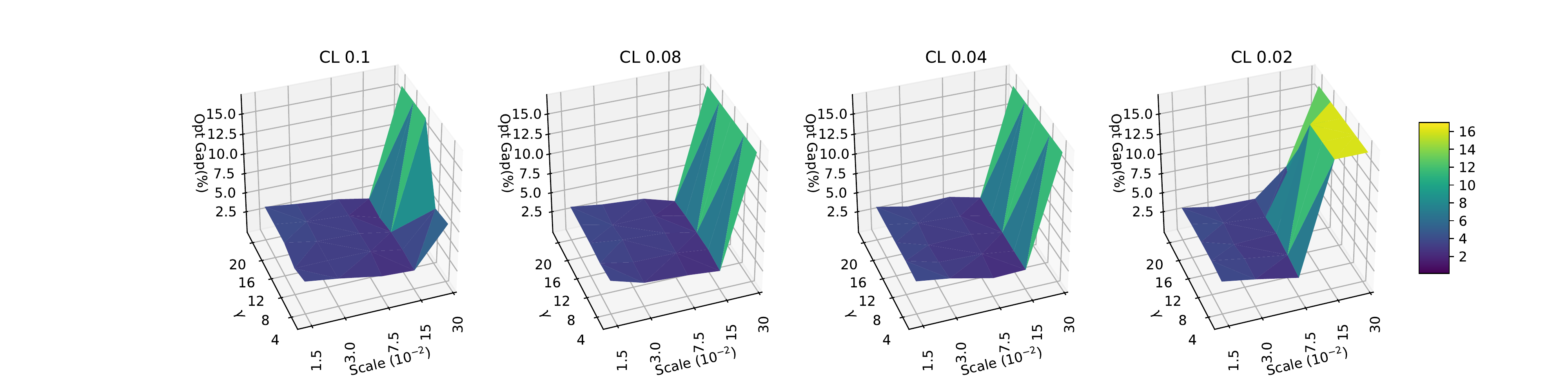}
\centering
\caption{Robustness Analysis for 8 region decomposition}
\label{fig:og8}
\end{figure*}
\begin{figure*}[!ht]
\includegraphics[trim=140 5 60 35,clip,width=\textwidth,keepaspectratio]{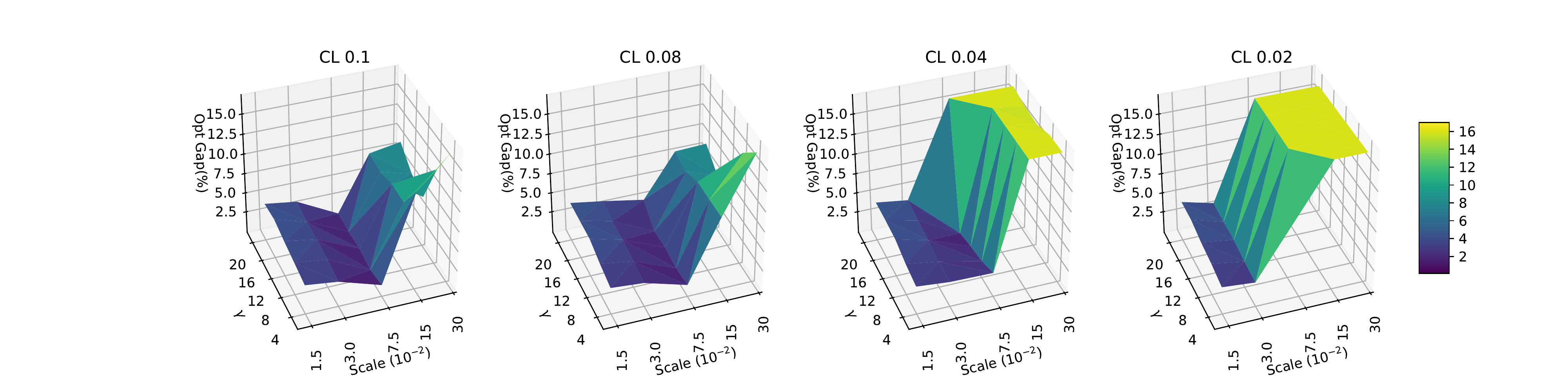}
\caption{Robustness Analysis for 12 region decomposition}
\label{fig:og12}
\end{figure*}

\begin{figure}
\centering
\begin{minipage}{.5\textwidth}
  \centering
   \includegraphics[width=\textwidth,keepaspectratio]{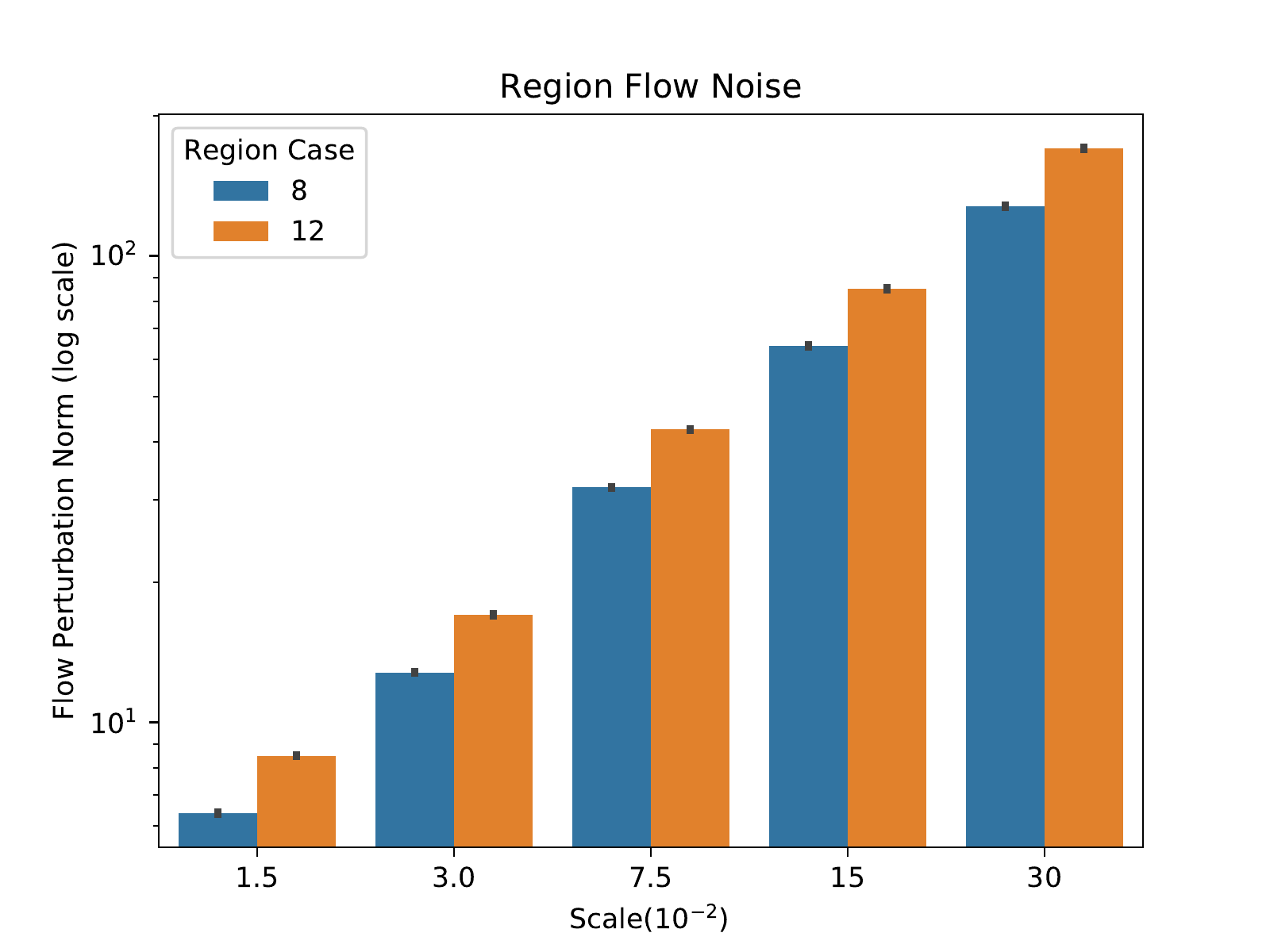}
\caption{Flow Noise Analysis}
\label{fig:fn_t}
\end{minipage}%
\begin{minipage}{.5\textwidth}
  \centering
   \captionsetup{justification=centering}
\includegraphics[width=\textwidth,keepaspectratio]{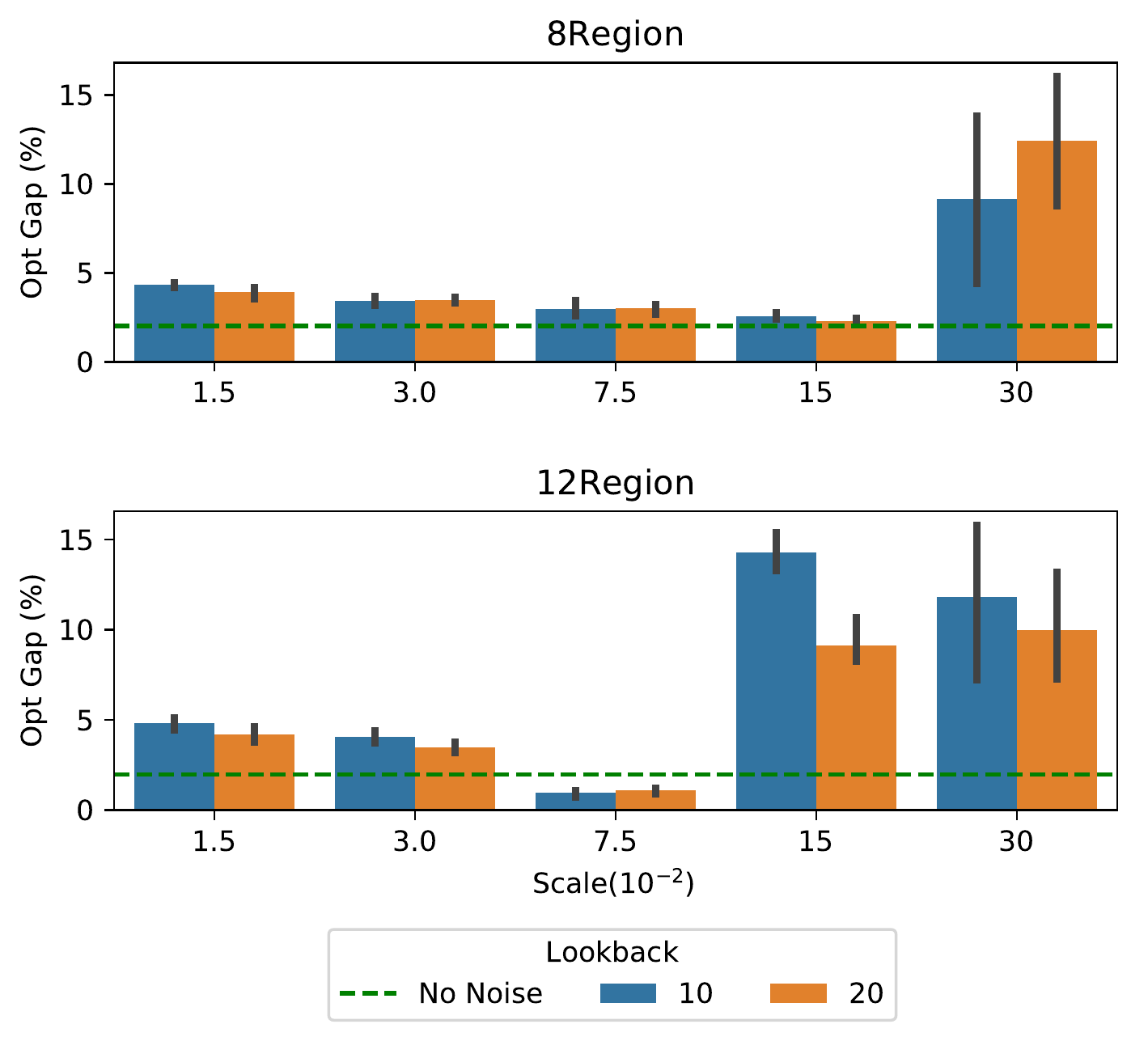}
\caption{Computational Analysis}
\label{fig:wwlim}
\end{minipage}
\end{figure}


\begin{table}
\centering
        \caption{EWMA Mixing parameter $\eta$}
        \begin{tabular}{|c|ccccc|}
        Scale  & \multicolumn{5}{c|}{$\gamma$} \\\cline{2-6}
        ($10^{-2}$)  & 4 & 8 & 12 & 16 & 20  \\ 
        \hline
        1.5    & 0.997 & 0.9976 & 0.9982 & 0.9988 & 0.9994 \\ 
        3.0    & 0.994 & 0.9952 & 0.9964 & 0.9976 & 0.9988\\
        7.5    & 0.985 & 0.988 & 0.991 & 0.994 & 0.997\\
        15.0   & 0.97 & 0.976 & 0.982 & 0.988 & 0.994\\
        30.0   & 0.94 & 0.952 & 0.964 & 0.976 & 0.988\\
        \end{tabular}
        \label{tab:eta1}
\end{table}
\section{Results}\label{sec:mresults}
In order to highlight the efficacy of our DP framework, we conduct numerous experiments on the 8 and 12 region decompositions of the IEEE 118 bus case. Our experiments revolve around analyzing the solution quality with respect to differing convergence limits, varying injected noise levels as well as changing the lookback window size for the control chart. We compute the solution quality  relative to a centralized, non DP formulation of the maintenance problem. 

\subsection{Experimental Setup}
Our decentralized implementation is based on the Message Passing Interface (MPI) which is a popular paradigm for distributed memory computation in the field of High Performance Computing (HPC). 
In order to implement Algorithm \ref{alg:syncd} each region in our problem was assigned to a single MPI process. Using a distributed memory model like MPI for communication helps us evaluate the algorithm in an environment close to the \emph{real-world}, where each region may represent individual participants of an ISO. Further, we impose an overlay network on top of the MPI layer which restricts a particular node to communicate only with the nodes representing the neighbors for its own region. The computational framework and software provided in this paper can be used as validation tools for large scale real world implementations on a myriad of computational platforms.

We used \verb|python| as the programming language for the framework. The \verb|mpi4py| \cite{mpi4py} package which is an MPI package for \verb|python| was used to build the decentralized framework. We used \verb|Gurobi 7.1| \cite{gurobi} for solving the MIQP problem represented by (\ref{eq:OPT}) on each node. We evaluate our model on the IEEE 118 bus case with data derived from the \verb|MATPOWER| library \cite{matpower}. We simulate a geographically dispersed set of regions on a high
performance cluster consisting of Intel Xeon CPUs with a clock rate of 2.80GHz with each core representing one region. Our planning horizon is of 1 week with hourly operational decisions and maintenance windows lasting 6 hours each. A preferred window for each generator scheduled for maintenance is provided as input to the local subproblem. We limit the maximum deviation from such preferences to at most 4 maintenance epochs. We impose a total runtime restriction of 10,600 secs for convergence. We compute the EWMA mixing parameter $\eta = 1.5\omega\times10^{-3}$ as represented in Table \ref{tab:eta1}.

The primal and dual tolerance limits for local convergence is given by $\beta_p, \beta_d = CL\cdot|\mathcal{B}_r|\cdot |T|$ respectively, where $CL$ is the convergence limit parameter. A higher $CL$ value could potentially yield faster convergence due to limited iterations that lead to reduced requirement for information exchange. Therefore, a higher $CL$ value might be a more preferred option given a reasonable degree of solution quality.



\subsection{Benchmark}
In order to benchmark our result, we consider a centralized version of Algorithm \ref{alg:syncd} without Differential Privacy. We use the centralized benchmark to rate the performance of the decentralized algorithm by measuring the relative optimality gap
\begin{equation}
Opt Gap = \frac{|\xi^{decent}-\xi^{cent}|}{\xi^{cent}}
\end{equation}
where $\xi^{decent},\xi^{cent}$ represent the total objective value upon convergence for the decentralized and centralized case respectively. Optimality gap values of experiments which did not converge have been capped at 16\%.




\subsection{Robustness Analysis}
Figures \ref{fig:og8}, \ref{fig:og12} depict surface plots pertaining to performance of our DP framework subject to different convergence limits. Each surface plots tracks the solution quality in terms of optimality gap with varying values of $\gamma$ and noise levels. 

From Figures \ref{fig:og8}, \ref{fig:og12} we observe that our DP framework provides stable performance, yielding optimality gap values of less than 5\% in most cases for both 8 and 12 region decompositions. We also observe that a higher $CL$ value does not necessarily come at the cost of poor optimality gap. Even with a high noise level, an appropriate selection of $\gamma$ might lead to an acceptable solution quality of around 5\% and 6\% in the 8 and 12 region case respectively with $CL$ value of 0.1. Further, as expected, the figures reveal that convergence becomes harder with higher noise levels and reduced $CL$ values. The figures also show that the number of regions in the network has a strong bearing on the convergence behavior and optimality gap. Such a behavior is expected since more regions imply a higher number of tie lines whose corresponding flow and phase angle values need to be balanced leading to greater difficulty in convergence. 

\subsection{Noise Analysis}
Figure \ref{fig:fn_t} depicts the 2- norm difference between the real flow and DP based flow values upon convergence for 8 and 12 region cases. We notice that with increasing scale of the noise as given on the x-axis, the flow perturbation value keeps increasing as well indicating our framework's success in preserving flow privacy. Further, the noise magnitude for each scale value does not vary much between the 8 and 12 region cases signaling its independent nature as well as applicability to different real world instances. 

\subsection{Lookback Analysis}
Figure \ref{fig:wwlim} depicts the trends observed for lookback sizes of 10 and 20 for a convergence limit of 0.1 in terms of a bar plot. We also compare the optimality gap obtained with varying noise values to the no noise case without DP. 

Overall, Figure \ref{fig:wwlim} shows that our framework performs remarkably well with varying noise levels compared to the no noise case. Moreover, we also observe that a lower lookback size leads to greater variance in optimality gap values. However, on the other hand, a lower lookback directly implies lesser number of iterations, leading to fewer message exchanges.

An interesting behavior in Figure \ref{fig:wwlim} is the decreasing optimality gap with increasing noise levels. This trend persists till $15\times 10^{-2}$ and $7.5 \times 10^{-2}$ noise levels for the 8 and 12 region cases respectively before increasing for greater noise levels. Such a behavior can be attributed to the sensitivity of ADMM to the penalty parameter $\rho$. Methods aimed at tuning the penalty parameter have been proposed  \cite{mhanna2018adaptive} for convex problems, however their applicability to Mixed Integer problems remains unexplored. Therefore, the penalty parameter was empirically chosen based on repeated trials with different values. Automatically adjusting the penalty parameter during run time for a finer performance of the Mixed Integer formulation is a key component of our future work. 
\subsection{Computational Analysis}
\begin{table}
\centering
        \caption{Computational Time (secs)}
        \begin{tabular}{|c|cc|cc|}
        Scale  & \multicolumn{2}{c|}{8 Regions} & \multicolumn{2}{c|}{12 Regions}\\\cline{2-5}
        ($10^{-2}$)  & Mean& Std. Dev. & Mean & Std. Dev.  \\ 
        \hline
        1.5    &1030.35 & 6.86 &  663.84 & 7.58\\ 
        3.0    &1023.20 & 14.18 & 665.88 & 7.12\\
        7.5    &1031.46 & 4.51 & 1990.82 & 2514.67\\
        15.0    &2240.02 & 3159.36 & 7052.10 & 4443.22\\
        30.0    &9217.68 & 3279.17 & 5516.14 & 4100.72\\
        \end{tabular}
        \label{tab:comp}
\end{table}

Table \ref{tab:comp} shows the mean computational time along with the standard deviation for 8 and 12 region cases for all CL and $\gamma$ values. The figures in Table \ref{tab:comp} also include cases where no convergence was observed within the maximum run time limit. 

We observe that for lower noise cases, the mean computation time proportionately decreases with increase in the number of regions from 8 to 12. Such a decrease is attributed to the computational speedup gained as a result of increased parallelism. However, we also observe an increase in the mean computational time for higher noise cases along with a variance. An increased variance and higher mean time to convergence is the consequence of two factors. First, due to a higher noise level, obtaining the Lagrangian balance might be more difficult leading to more iterations and prolonged convergence time. Additionally, higher noise levels result in greater likelihood of not observing convergence within the prescribed time limit, ultimately increasing both the mean and variance. 

\section{Conclusion}\label{sec:mconclusion}
In this paper we present a differential privacy driven approach for solving planning problems in a decentralized fashion. We choose the short term maintenance and commitment problem as our target for demonstrating the efficacy of our approach due to its practical and critical aspects. Our decentralization is driven by a region based decomposition representing real world utility stakeholders. We obtain a decentralized formulation by dualizing the phase angles and flow constraints neighboring regions and iteratively balance these using ADMM. For orchestrating differential privacy, we exploit the linear relationship between the flow and phase angles. By injecting carefully coordinated exponential noise on the phase angles, we derive strong privacy guarantees on the flow values. In order to improve convergence, we adopt an EWMA based consensus averaging strategy in addition to a CLT based control chart. Our consensus strategy coupled with our control chart mechanism leads to a stable superior convergence for fairly large noise values. Further, using our HPC implementation of our decentralized framework, we show good solution quality that rivals that of the centralized as well as the no noise benchmarks. 
\bibliography{artifact}

\begin{thebibliography}{10}

\bibitem{chowdhury1990review}
B.~H. Chowdhury and S.~Rahman, ``A review of recent advances in economic
  dispatch,'' {\em IEEE transactions on power systems}, vol.~5, no.~4,
  pp.~1248--1259, 1990.

\bibitem{capitanescu2011state}
F.~Capitanescu, J.~M. Ramos, P.~Panciatici, D.~Kirschen, A.~M. Marcolini,
  L.~Platbrood, and L.~Wehenkel, ``State-of-the-art, challenges, and future
  trends in security constrained optimal power flow,'' {\em Electric Power
  Systems Research}, vol.~81, no.~8, pp.~1731--1741, 2011.

\bibitem{unit_commitment}
N.~P. Padhy, ``Unit commitment-a bibliographical survey,'' {\em IEEE
  Transactions on Power Systems}, vol.~19, pp.~1196--1205, May 2004.

\bibitem{maintDereg}
R.~B. and, ``Composite system maintenance coordination in a deregulated
  environment,'' {\em IEEE Transactions on Power Systems}, vol.~20,
  pp.~485--492, Feb 2005.

\bibitem{maintDereg2}
C.~Feng and X.~Wang, ``A competitive mechanism of unit maintenance scheduling
  in a deregulated environment,'' {\em IEEE transactions on power systems},
  vol.~25, no.~1, pp.~351--359, 2010.

\bibitem{marwali1998integrated}
M.~Marwali and S.~Shahidehpour, ``Integrated generation and transmission
  maintenance scheduling with network constraints,'' {\em IEEE Transactions on
  Power Systems}, vol.~13, no.~3, pp.~1063--1068, 1998.

\bibitem{yang2013consensus}
S.~Yang, S.~Tan, and J.-X. Xu, ``Consensus based approach for economic dispatch
  problem in a smart grid,'' {\em IEEE Transactions on Power Systems}, vol.~28,
  no.~4, pp.~4416--4426, 2013.

\bibitem{javad}
M.~J. Feizollahi, M.~Costley, S.~Ahmed, and S.~Grijalva, ``Large-scale
  decentralized unit commitment,'' {\em International Journal of Electrical
  Power \& Energy Systems}, vol.~73, pp.~97 -- 106, 2015.

\bibitem{ramanan2017asynchronous}
P.~Ramanan, M.~Yildirim, E.~Chow, and N.~Gebraeel, ``Asynchronous decentralized
  framework for unit commitment in power systems,'' {\em Procedia Computer
  Science}, vol.~108, pp.~665--674, 2017.

\bibitem{asynch2019uc}
P.~{Ramanan}, M.~{Yildirim}, E.~{Chow}, and N.~{Gebraeel}, ``An asynchronous,
  decentralized solution framework for the large scale unit commitment
  problem,'' {\em IEEE Transactions on Power Systems}, vol.~34, pp.~3677--3686,
  Sep. 2019.

\bibitem{xavier2020decomposable}
{\'A}.~S. Xavier, F.~Qiu, and S.~S. Dey, ``Decomposable formulation of
  transmission constraints for decentralized power systems optimization,'' {\em
  arXiv preprint arXiv:2001.07771}, 2020.

\bibitem{dwork2014algorithmic}
C.~Dwork, A.~Roth, {\em et~al.}, ``The algorithmic foundations of differential
  privacy,'' {\em Foundations and Trends{\textregistered} in Theoretical
  Computer Science}, vol.~9, no.~3--4, pp.~211--407, 2014.

\bibitem{cortes2016differential}
J.~Cort{\'e}s, G.~E. Dullerud, S.~Han, J.~Le~Ny, S.~Mitra, and G.~J. Pappas,
  ``Differential privacy in control and network systems,'' in {\em 2016 IEEE
  55th Conference on Decision and Control (CDC)}, pp.~4252--4272, IEEE, 2016.

\bibitem{wang2016stochastic}
Y.~Wang, Z.~Li, M.~Shahidehpour, L.~Wu, C.~Guo, and B.~Zhu, ``Stochastic
  co-optimization of midterm and short-term maintenance outage scheduling
  considering covariates in power systems,'' {\em IEEE Transactions on Power
  Systems}, vol.~31, no.~6, pp.~4795--4805, 2016.

\bibitem{admm_boyd}
S.~Boyd, N.~Parikh, E.~Chu, B.~Peleato, and J.~Eckstein, ``Distributed
  optimization and statistical learning via the alternating direction method of
  multipliers,'' {\em Found. Trends Mach. Learn.}, vol.~3, pp.~1--122, Jan.
  2011.

\bibitem{liu2018decentralized}
J.~Liu, H.~Cheng, P.~Zeng, L.~Yao, C.~Shang, and Y.~Tian, ``Decentralized
  stochastic optimization based planning of integrated transmission and
  distribution networks with distributed generation penetration,'' {\em Applied
  Energy}, vol.~220, pp.~800--813, 2018.

\bibitem{ghazvini2012coordination}
M.~F. Ghazvini, H.~Morais, and Z.~Vale, ``Coordination between mid-term
  maintenance outage decisions and short-term security-constrained scheduling
  in smart distribution systems,'' {\em Applied energy}, vol.~96, pp.~281--291,
  2012.

\bibitem{fu2007security}
Y.~Fu, M.~Shahidehpour, and Z.~Li, ``Security-constrained optimal coordination
  of generation and transmission maintenance outage scheduling,'' {\em IEEE
  Transactions on Power Systems}, vol.~22, no.~3, pp.~1302--1313, 2007.

\bibitem{fu2009coordination}
Y.~Fu, Z.~Li, M.~Shahidehpour, T.~Zheng, and E.~Litvinov, ``Coordination of
  midterm outage scheduling with short-term security-constrained unit
  commitment,'' {\em IEEE Transactions on Power Systems}, vol.~24, no.~4,
  pp.~1818--1830, 2009.

\bibitem{muratp1}
M.~Yildirim, X.~A. Sun, and N.~Z. Gebraeel, ``Sensor-driven condition-based
  generator maintenance scheduling part i: Maintenance problem,'' {\em IEEE
  Transactions on Power Systems}, vol.~PP, no.~99, pp.~1--10, 2016.

\bibitem{muratp2}
M.~Yildirim, X.~A. Sun, and N.~Z. Gebraeel, ``Sensor-driven condition-based
  generator maintenance scheduling part ii: Incorporating operations,'' {\em
  IEEE Transactions on Power Systems}, vol.~PP, no.~99, pp.~1--9, 2016.

\bibitem{sun2018cyber}
C.-C. Sun, A.~Hahn, and C.-C. Liu, ``Cyber security of a power grid:
  State-of-the-art,'' {\em International Journal of Electrical Power \& Energy
  Systems}, vol.~99, pp.~45--56, 2018.

\bibitem{zhang2016inclusion}
Y.~Zhang, L.~Wang, Y.~Xiang, and C.-W. Ten, ``Inclusion of scada cyber
  vulnerability in power system reliability assessment considering optimal
  resources allocation,'' {\em IEEE Transactions on Power Systems}, vol.~31,
  no.~6, pp.~4379--4394, 2016.

\bibitem{zhou2019differential}
F.~Zhou, J.~Anderson, and S.~H. Low, ``Differential privacy of aggregated dc
  optimal power flow data,'' in {\em 2019 American Control Conference (ACC)},
  pp.~1307--1314, IEEE, 2019.

\bibitem{fioretto2018constrained}
F.~Fioretto and P.~Van~Hentenryck, ``Constrained-based differential privacy:
  Releasing optimal power flow benchmarks privately,'' in {\em International
  Conference on the Integration of Constraint Programming, Artificial
  Intelligence, and Operations Research}, pp.~215--231, Springer, 2018.

\bibitem{fioretto2019differential}
F.~Fioretto, T.~W. Mak, and P.~Van~Hentenryck, ``Differential privacy for power
  grid obfuscation,'' {\em IEEE Transactions on Smart Grid}, 2019.

\bibitem{mak2019privacy}
T.~W. Mak, F.~Fioretto, L.~Shi, and P.~Van~Hentenryck, ``Privacy-preserving
  power system obfuscation: A bilevel optimization approach,'' {\em IEEE
  Transactions on Power Systems}, 2019.

\bibitem{fioretto2019privacy}
F.~Fioretto, T.~W. Mak, and P.~Van~Hentenryck, ``Privacy-preserving obfuscation
  of critical infrastructure networks,'' {\em arXiv preprint arXiv:1905.09778},
  2019.

\bibitem{mak2019privacy2}
T.~W. Mak, F.~Fioretto, and P.~Van~Hentenryck, ``Privacy-preserving obfuscation
  for distributed power systems,'' {\em arXiv preprint arXiv:1910.04250}, 2019.

\bibitem{kotz2012laplace}
S.~Kotz, T.~Kozubowski, and K.~Podgorski, {\em The Laplace distribution and
  generalizations: a revisit with applications to communications, economics,
  engineering, and finance}.
\newblock Springer Science \& Business Media, 2012.

\bibitem{mpi4py}
L.~Dalcín, R.~Paz, M.~Storti, and J.~D’Elía, ``Mpi for python: Performance
  improvements and mpi-2 extensions,'' {\em Journal of Parallel and Distributed
  Computing}, vol.~68, no.~5, pp.~655 -- 662, 2008.

\bibitem{gurobi}
I.~Gurobi~Optimization, ``Gurobi optimizer reference manual,'' 2015.

\bibitem{matpower}
R.~D. Zimmerman, C.~E. Murillo-Sanchez, and R.~J. Thomas, ``Matpower:
  Steady-state operations, planning, and analysis tools for power systems
  research and education,'' {\em IEEE Transactions on Power Systems}, vol.~26,
  pp.~12--19, Feb 2011.

\bibitem{mhanna2018adaptive}
S.~Mhanna, G.~Verbi{\v{c}}, and A.~C. Chapman, ``Adaptive admm for distributed
  ac optimal power flow,'' {\em IEEE Transactions on Power Systems}, vol.~34,
  no.~3, pp.~2025--2035, 2018.

\end{thebibliography}
\bibliographystyle{ieeetr}
\end{document}